\documentclass{ifacconf}

\usepackage{amsmath,amssymb,amsfonts}
\usepackage{algorithm}
\usepackage{algpseudocode}
\usepackage{subcaption}
\usepackage{graphicx}
\usepackage{xfrac}
\usepackage{enumerate}
\usepackage{textcomp}
\usepackage{tikz} 
\usepackage{pgfplots}
\usepackage{xfrac}
\usepackage{mathptmx}
\usepackage{booktabs}
\usepackage{array}
\usepackage{arydshln}
\usepackage{mathtools} 
\DeclareMathAlphabet{\mathcal}{OMS}{cmsy}{m}{n}
\usepackage{bm}
\usetikzlibrary{positioning}
\usetikzlibrary{calc}
\usetikzlibrary{decorations.markings,patterns,decorations.pathmorphing}
\usetikzlibrary{arrows, shapes}
\tikzstyle{block} = [draw, thick, node distance=0.5cm, minimum width=1cm, inner sep=6pt]
\tikzstyle{sum} = [draw, thick, circle, node distance=1cm, inner sep=3.5pt, path picture={\node at (path picture bounding box.center) [draw, anchor = center] {$+$};}]
\usepackage{graphicx}      
\usepackage{natbib}        

\newcommand{\N}{\mathbb{N}}

\newcommand{\x}{x}
\newcommand{\p}{p}
\newcommand{\w}{w}
\renewcommand{\u}{u}
\newcommand{\q}{q}
\newcommand{\z}{z}

\newcommand{\filty}{s}
\renewcommand{\t}{t}
\renewcommand{\k}{k}
\newcommand{\A}[1]{A_{#1}}
\newcommand{\B}[2]{B_{#1}^{#2}}
\newcommand{\C}[2]{C_{#1}^{#2}}
\newcommand{\D}[3]{D_{#1}^{#2#3}}

\newcommand{\G}{G}

\newcommand{\DG}{\Delta \star \G}

\newcommand{\filt}{\Psi}
\newcommand{\all}{\Sigma}

\DeclareMathOperator{\diag}{diag}

\newcommand{\symb}[1]{\begin{bmatrix}\vphantom{#1}\bullet \end{bmatrix}^{\!\top}\! #1 }

\newcommand{\nx}{{n_\x}}
\let\greeknu\nu
\renewcommand{\nu}{{n_\u}}
\newcommand{\nw}{{n_\w}}

\newcommand{\nq}{{n_\q}}
\newcommand{\np}{{n_\p}}

\newcommand{\nz}{{n_\z}}

\newcommand{\nfilty}{{n_\filty}}

\newcommand{\nfilt}{{n_\filtx}}

\newcommand{\filtx}{\psi}

\newcommand{\M}{M}
\newcommand{\X}{X}
\renewcommand{\P}{P}

\newcommand{\R}{\mathbb{R}}
\renewcommand{\S}{\mathbb{S}}

\newcommand{\RHinf}{\mathbb{RH}_\infty}

\newcommand{\ssrep}[4]{\left[\begin{array}{c|c} #1 & #2 \\ \hline #3 & #4\end{array}\right]}
\newcommand{\ssrepflex}[4]{\left[\begin{array}{#1|#2} #3 \\[0.05cm] \hline\rule{0pt}{2.45ex} #4\end{array}\right]}

\newcommand{\peak}{{\mathrm{peak}}}
\newcommand{\indpeak}{{\mathrm{peak\text{-}ind}}}
\newcommand{\Deltaset}{\boldsymbol{\Delta}}
\newcommand{\Mset}{\boldsymbol{\M}}
\newcommand{\MXset}{\boldsymbol{\M\X}}

\newcommand{\diagmat}[1]{{\arraycolsep=1pt\begin{bmatrix} #1\end{bmatrix}}}
\newenvironment{proof}{\begin{pf}}{\,\hfill $\square$\end{pf}}
\usepackage{calc}

\renewcommand{\refeq}[2]{\overset{\makebox[0pt][c]{\scriptsize #1}}{#2}}

\begin{document}
\begin{frontmatter}

\title{Robust peak-to-peak gain analysis\\ using integral quadratic constraints\thanksref{footnoteinfo}} 

\thanks[footnoteinfo]{
  F. Allgöwer and M. A. Müller are thankful that this work was funded by Deutsche Forschungsgemeinschaft (DFG, German Research Foundation) – AL 316/12-2 and MU 3929/1-2 - 279734922.  L. Schwenkel thanks the International Max Planck Research School for Intelligent Systems (IMPRS-IS) for supporting him.}

\author[First]{Lukas Schwenkel} 
\author[Second]{Johannes Köhler} 
\author[Third]{Matthias A. Müller}
\author[First]{Frank Allgöwer}

\address[First]{University of Stuttgart, Institute for Systems Theory and Automatic Control,\\ \ \hspace{-0.7cm} 70550 Stuttgart, Germany (e-mail: $\{$schwenkel, allgower$\}$@ist.uni-stuttgart.de)\hspace{-0.8cm}\ }
\address[Second]{ETH Zurich, Institute for Dynamical Systems and Control, ZH-8092, Switzerland (e-mail: jkoehle@ethz.ch)}
\address[Third]{Leibniz University Hannover, Institute of Automatic Control, 30167 Hannover, Germany (email: mueller@irt.uni-hannover.de)}

\begin{abstract}
  This work provides a framework to compute an upper bound on the robust peak-to-peak gain of discrete-time uncertain linear systems using integral quadratic constraints (IQCs).
Such bounds are of particular interest in the computation of reachable sets and the $\ell_1$-norm, as well as when safety-critical constraints need to be satisfied pointwise in time.
The use of $\rho$-hard IQCs with a terminal cost enables us to deal with a wide variety of uncertainty classes, for example, we provide $\rho$-hard IQCs with a terminal cost for the class of parametric uncertainties.
This approach unifies, generalizes, and significantly improves state-of-the-art methods, which is also demonstrated in a numerical example.
\end{abstract}


\end{frontmatter}

\section{Introduction}
This work provides a framework to guarantee an upper bound on the robust peak-to-peak gain of discrete-time uncertain linear systems.
Such bounds are particularly interesting in safety critical systems where one needs to guarantee constraint satisfaction pointwise in time despite persistent disturbances, uncertainties, and noise.
The control community has been interested in peak-to-peak gains and the closely related $\ell_1$-norm since the early works of~\cite{Vidyasagar1986} and~\cite{Dahleh1987}.
Ten years later,~\cite{Abedor1996} proposed a computationally scalable (but approximate) solution based on linear matrix inequalities (LMIs).
All these works address the nominal problem without model uncertainties.
This restriction has been addressed by~\cite{Ji2007} and~\cite{Rieber2008}, where a robust bound on the peak-to-peak gain is computed for linear systems with parametric uncertainties.
In comparison, we propose a framework based on integral quadratic constraints (IQCs), which enables us to deal with a wide variety of uncertainty classes.
Furthermore, using IQCs we can reduce conservatism by exploiting additional knowledge on the structure or time-invariance of the uncertainty, and we can recover the results from~\cite{Ji2007} and~\cite{Rieber2008} as special cases.

IQCs were first introduced by~\cite{Megretski1997} and have proven to be an efficient tool to analyze uncertain systems (cf. the tutorial by~\cite{Veenman2016}).
In the literature, a distinction is made between \emph{hard} IQCs, which have to hold on all finite horizons $[0,T]$, $T\geq 0$, and \emph{soft} IQCs, which only have to hold over the infinite horizon $T\to \infty$.
\cite{Veenman2013} already conjectured that hard time-bounds, such as the desired peak-to-peak gain, can be derived by assuming hard IQCs and using a dissipativity-based proof.
However, hard IQCs are more restrictive than soft IQCs and thus often come with additional conservatism.
Hence,~\cite{Scherer2018} proposed to relax hard IQCs by using a \emph{terminal cost} and thereby removing some of the conservatism (see also~\citep{Scherer2022a} and~\citep{Scherer2022}).
Moreover, to bound the impact of persistent disturbances of possibly infinite energy,~\citet{Abedor1996},~\cite{Ji2007},~\cite{Rieber2008} used exponential stability bounds.
Similarly, exponential stability analysis within the IQC framework can be performed using $\rho$-hard IQCs as proposed by \cite{Lessard2016}.
Considering the above discussion, we utilize $\rho$\emph{-hard IQCs with a terminal cost} to analyze the robust peak-to-peak gain in a general setting.

\emph{Related work.} \cite{Jaoude2020} use IQCs to analyze the robust energy-to-peak gain, i.e., pointwise bound on the output given a bound on the energy of the disturbance signal.
Similarly,~\cite{Scherer2022} provides pointwise bounds on the output of an uncertain system described with IQCs assuming no disturbance but a norm-bound on the initial condition.
Furthermore, outer approximations of the reachable set have been provided by~\cite{Yin2020} and~\cite{Buch2021} using the IQC framework and assuming disturbances of finite energy.
All these works, however, do not allow for persistent disturbances. 
In the context of IQCs, only~\cite{Jaoude2021} deal with persistent disturbances and provide an outer approximation of the reachable set.
However, their approach requires pointwise IQCs whereas the proposed approach requires the less restrictive $\rho$-hard IQCs with a terminal cost.
We demonstrate in Example~\ref{exmp:jaude} that this can lead to less conservative results.

The results of this work are in particular relevant for the design of robust model predictive control (MPC) schemes, where constraints in terms of pointwise bounds on the output need to be satisfied despite uncertainties in the system and persistent disturbances.
In our previous works~\citep{Schwenkel2020} and~\citep{Schwenkel2022a}, we proposed a robust MPC scheme based on an outer approximation of the reachable set using $\rho$-hard IQCs.
In this work, we improve the approach to obtain outer approximations therein by allowing for $\rho$-hard IQCs \emph{with a terminal cost}.
Moreover, the proposed robust peak-to-peak gain analysis can be used in MPC to directly compute a suitable constraint tightening to ensure constraint satisfaction.
Even when the constraints are directly influenced by the uncertainty.

\emph{Outline.} After defining the problem setup in Section~\ref{sec:setup}, we present several contributions in this work.
In Section~\ref{sec:ana}, we provide a framework to give guaranteed bounds on the robust peak-to-peak gain of an uncertain system where the uncertain components are characterized by $\rho$-hard IQCs with a terminal cost. 
In addition, we also show how to compute an outer approximation of the reachable set. 
In Section~\ref{sec:param}, we provide $\rho$-hard IQCs with a terminal cost for the special case of parametric uncertainties and show that existing results are special cases of our approach. 
Finally, in Section~\ref{sec:exmp}, we demonstrate with two examples that our results have little conservatism and can significantly improve state-of-the-art methods.

\emph{Notation.}
We denote the set of eigenvalues of a matrix $A\in \R^{n\times n}$ by $\lambda(A)\subseteq \mathbb C$.
For $x\in\R^n$, denote the infinity norm by $\|x\|_\infty = \max_{i} |x_i|$ and the Euclidean norm by $\|x\| = \sqrt{x^\top x}$.
The set of (bounded) sequences $x:\N\to\R^n$ is denoted by $\ell_{2\mathrm{e}}^n$ ($\ell_{\infty}^n$) and the $\ell_{\infty}$-norm of a sequence $x\in\ell_{\infty}^n$ is denoted by $\|x\|_\infty = \sup_{t\in\N} \|x_t\|_\infty$.
For matrices $A$, $B$, $C$, $D$ with suitable dimensions we denote the operator (system) that maps input signals $u \in \ell_{2\mathrm{e}}^m$ to output signal $y \in \ell_{2\mathrm{e}}^n$ according to $x_{k+1} = Ax_k+Bu_k$, $x_0=0$, $y_k = Cx_k +D u_k$ for all $k\geq0$ by $y=\ssrep{A}{B}{C}{D}u$. 
We denote the set of exponentially stable systems with input dimension $m$ and output dimension $n$ by $\RHinf^{n\times m}:=\left\{\left.\ssrep{A}{B}{C}{D}\right| \max_{\lambda \in \lambda(A)}  |\lambda|<1 \right\}$.
The set of symmetric matrices $A=A^\top \in \R^{n\times n}$ is denoted by $\S^n$.
If $A \in\S^{n}$ is a positive (semi-)definite matrix, we write $A\succ 0$ ($A\succeq 0$).
If $A \in\S^{n}$ is a negative (semi-)definite matrix, we write $A\prec 0$ ($A\preceq 0$).
If $x\in\R^n$ and $A\succ 0$, we write $\|x\|_A = \sqrt{x^\top Ax}$ to indicate that this is a weighted norm on $\R^n$.
For matrices $A\in \R^{n\times m}$ and $P \in \S^{n}$, we denote $\symb{PA}=A^\top P A$ and call $P$ the inner and $A$ the outer factor.
For a vector $x\in\R^n$ we define the diagonal matrix $\diag(x)=\begin{bmatrix*}[l] x_1\\[-2mm]\hspace{10pt} \ddots\\[-1mm]\hspace{20pt}x_n\end{bmatrix*}$.
We denote the convex hull of a set of points $\delta^1,\dots,\delta^m\in\R^n$ by $\mathrm{conv}\{\delta^1,\dots,\delta^m\}$.

\section{Setup}\label{sec:setup}
We consider the linear system $\G$ given by
\begin{subequations}\label{eq:sys}
\begin{alignat}{4}
  \makebox[0pt]{$\G$:\qquad\qquad\ }\x_{\k+1} &= \A\G \x_\k &&+ \B\G\p \p_\k &&+ \B\G\w \w_\k \\
  \q_\k &= \C\G\q \x_\k  &&+ \D\G\q\p \p_\k &&+ \D\G\q\w \w_\k \\
  \z_\k &= \C\G\z \x_\k  &&+ \D\G\z\p \p_\k &&+ \D\G\z\w \w_\k
\end{alignat}
\end{subequations}
where $\k\in\N$ denotes the time index and the system is initialized at $\x_0=0$. 
The channel $\p\to\q$ in $\G$ is called the uncertainty channel and $\w\to\z$ the performance channel.
The dimensions of the signals are $\x\in \ell_{2\mathrm{e}}^\nx$, $\p\in\ell_{2\mathrm{e}}^\np$, $\q \in \ell_{2\mathrm{e}}^\nq$, $\w\in\ell_{2\mathrm{e}}^\nw$, and $\z\in\ell_{2\mathrm{e}}^\nz$ and the system matrices have suitable dimensions. 
The uncertainty channel of system \eqref{eq:sys} is in feedback with a causal unknown operator $\Delta:\ell_{2\mathrm{e}}^\nq\to\ell_{2\mathrm{e}}^\np$ that belongs to a known set $\Delta\in\Deltaset$, i.e., 
\begin{align}\label{eq:Delta}
  \p = \Delta(\q).
\end{align}
We denote this feedback interconnection by $\DG$ and throughout this work, we assume well-posedness, i.e., for all $\w\in\ell_{2\mathrm{e}}$ and $\Delta\in\Deltaset$ there is a unique solution of~\eqref{eq:sys},~\eqref{eq:Delta} that causally depends on $\w$.
This work analyzes the peak-to-peak gain of the performance channel, i.e.,
\begin{align}\label{eq:def_indpeak}
  \|\DG\|_\indpeak = \sup_{\w\in\ell_{\infty}^\nw, \w\neq 0} \frac{\|\z\|_\peak}{\|\w\|_\peak},
\end{align}
where the peak-norm is defined by
\begin{align}
  \|\z\|_\peak = \sup_{\t\in\N }  \|\z_\t\|.
\end{align}
Since $\|\DG\|_\indpeak$ depends on the uncertainty $\Delta$, we are particularly interested in desirably small upper bound $\gamma>0$ on the robust peak-to-peak gain
\begin{align}\label{eq:peak}
  \sup_{\Delta \in \Deltaset}\|\DG \|_\indpeak \leq \gamma.
\end{align}

\begin{rem}
  As discussed by~\cite{Rieber2008}, the peak-to-peak gain is closely related to the $\ell_{\infty}$-to-$\ell_{\infty}$ gain, also known as $\ell_1$-norm. In particular, for every system $H\in\RHinf^{n\times m }$ it holds 
  \begin{align*}
    \frac{1}{\sqrt{n}}\|H \|_\indpeak \leq \|H\|_1 \leq \sqrt{m} \|H \|_\indpeak,
  \end{align*}
  such that every bound $\gamma$ on the peak-to-peak gain implies also a bound on the $\ell_1$-norm.
\end{rem}

We assume that we can describe the uncertainty set $\Deltaset$ by a $\rho$-hard IQCs with a terminal cost.

\begin{defn}[$\rho$-hard IQC with terminal cost]\label{defn:iqc}
  Let $\rho\in (0,1)$, $\M\in\S^{\nfilty}$, $\X\in\S^{\nfilt}$, and $\filt \in \RHinf^{\nfilty\times (\nq+\np)}$ with state space realization
  \begin{subequations}\label{eq:allsys_start}
    \begin{align}
      \makebox[0pt]{$\filt$:\qquad\qquad\ }\filtx_{\k+1} &= \A\filt \filtx_\k + \B\filt\q \q_\k  + \B\filt\p \p_\k\\
      \filty_\k &= \C\filt\filty \filtx_\k + \D\filt\filty\q \q_\k + \D\filt\filty\p \p_\k
    \end{align}
  \end{subequations}
  with $\filtx_0=0\in\R^{\nfilt}$. A causal operator $\Delta:\ell_{2\mathrm{e}}^\nq\to\ell_{2\mathrm{e}}^\np$ is said to satisfy the
  \begin{itemize}
    \item $\rho$-hard IQC defined by $(\rho,\filt,\M,\X)$ if for $\p=\Delta(\q)$ and for all $\q\in\ell_{2\mathrm{e}}^\nq$, $\t\in\N$ it holds
    \begin{align}\label{eq:hardIQC}
      \sum_{\k=0}^{\t} \rho^{\t-\k}\filty_\k^\top \M\filty_\k + \filtx_{\t+1}^\top \X \filtx_{\t+1} \geq 0.
    \end{align}
    \item pointwise IQC defined by $(\filt,\M)$ if for $\p=\Delta(\q)$ and for all $\q\in\ell_{2\mathrm{e}}^\nq$, $\k\in\N$ it holds $\filty_\k^\top \M\filty_\k\geq 0$.
  \end{itemize} 
\end{defn}
We call $\M$ the multiplier, $\filt$ the filter, $\rho$ the squared exponential decay rate\footnote{In literature on $\rho$-hard IQCs, $\rho$ is usually the exponential decay rate and not its square. However, using $\rho$ instead of $\rho^2$ simplifies notation significantly.}, $\X$ the terminal cost matrix, and $\filty$ the filter output. 

\begin{rem}
  Definition~\ref{defn:iqc} reduces to the standard definition of $\rho$-hard IQCs (without a terminal cost) by~\cite{Lessard2016} if $\X=0$.
  As noted in~\cite{Lessard2016}, a pointwise IQC defined by $(\filt,\M)$ implies a $\rho$-hard IQC defined by $(\rho,\filt,\M,0)$ for all $\rho\in (0,1)$. 
\end{rem}

\section{Robust analysis}\label{sec:ana}
In this section, we present our main results: a guaranteed upper bound on the peak-to-peak gain and a guaranteed outer approximation of the reachable set. 
For the analysis, we augment the system $\G$ with the filter $\filt$ to obtain the following augmented system $\all$ with state $\chi=\begin{bmatrix}\filtx \\ \x \end{bmatrix}\in\R^{n_\chi}$, state space representation
\begin{subequations}\label{eq:aug_sys}
  \begin{align}
    \makebox[0pt]{$\all$:\qquad\qquad\ }
    \chi_{\k+1} &= \A\all \chi_\k + \B\all\p \p_\k + \B\all\w \w_\k \\
    \filty_\k &= \C\all\filty \chi_\k + \D\all{\filty}\p \p_\k + \D\all{\filty}\w \w_\k\\
    \z_\k &= \C\all\z \chi_\k  + \D\all\z\p \p_\k + \D\all\z\w \w_\k,
  \end{align}
\end{subequations}
initial condition $\chi_0=0$, and the matrices 
\begin{align*}
  &\left[\begin{array}{c|c:c}\A\all & \B\all\p & \B\all\w \\[0.5mm] \hline\rule{0pt}{3.5mm} \C\all\filty & \D\all{\filty}\p & \D\all{\filty}\w\\[0.5mm] \hdashline\rule{0pt}{3.5mm} \C\all\z & \D\all\z\p & \D\all\z\w \end{array}\right]= \left[\begin{array}{cc|c:c} 
    \A\filt & \B\filt\q\C\G\q & \B\filt\p+ \B\filt\q\D\G\q\p &\B\filt\q\D\G\q\w\\ 
    0 & \A\G & \B\G\p & \B\G\w \\[0.03cm] \hline\rule{0pt}{2ex}
    \C\filt\filty & \D\filt\filty\q\C\G\q & \D\filt\filty\p+ \D\filt\filty\q\D\G\q\p & \D\filt\filty\q\D\G\q\w \\[0.03cm] \hdashline\rule{0pt}{2ex}
    0 & \C\G\z & \D\G\z\p & \D\G\z\w\end{array}\right].
\end{align*}
This reformulation allows us to state the following result.
\begin{thm}\label{thm:ana}
  Assume that all $\Delta\in\Deltaset$ satisfy the $\rho$-hard IQC defined by $(\rho,\filt,\M,\X)$. Further, assume that there exist $\P\in\S^{n_\chi}$, $\gamma\geq \mu$, and $\mu\geq 0$ such that
  \begin{align}\label{eq:ana_LMI}
    \symb {\diagmat{-\rho \P\\&\P\\&&\M\\&&&- \mu I}
      \begin{bmatrix}
        I &  0 & 0\\
        \A\all & \B\all\p & \B\all\w \\[0.5mm]
        \C\all\filty & \D\all{\filty}\p & \D\all{\filty}\w \\
         0& 0&I
       \end{bmatrix}}\preceq 0,
  \end{align}
and
 \begin{align}\label{eq:ana_LMI2}
   \newcommand{\matspace}{\hspace{0.4cm}}
    \symb{{\arraycolsep=1pt\begin{bmatrix}-\rho\P\\
        &\tilde \X \\
        &&\M\\
        &&&\frac {\rho}{\gamma(1-\rho)} I\\
        &&&&-\frac{\rho(\gamma-\mu) }{1-\rho}I\end{bmatrix}}
    \begin{bmatrix}I&0&0\\[0.3mm]
      \A\all & \B\all\p & \B\all\w \\[0.8mm]
      \C\all\filty&\D\all\filty\p&\D\all\filty\w\\[0.8mm]
      \C\all\z & \D\all\z\p & \D\all\z\w \\[0.5mm]
      0&0&I\end{bmatrix}}  \preceq 0,
\end{align}
  hold  with $\tilde \X = \diagmat{\X \\ & 0}$. Then $\|\DG\|_\indpeak \leq \gamma $ for all $\Delta \in \Deltaset$.
\end{thm}
\begin{proof}
  Multiplying both matrix inequalities~\eqref{eq:ana_LMI} and~\eqref{eq:ana_LMI2} from the left by $\begin{bmatrix}\chi_\k^\top & \p_\k^\top & \w_\k^\top\end{bmatrix}$ and from the right by its transpose, and using~\eqref{eq:aug_sys}, we obtain
  \begin{align}\label{eq:ana_diss1}
    \delta_1(k):=\chi_{\k+1}^\top \P\chi_{\k+1}-\rho\chi_\k^\top \P\chi_\k+\filty_\k^\top\M\filty_\k - \mu \|\w_\k\|^2 \leq 0
  \end{align}
  and
  \begin{align}\nonumber 
    \delta_2(k)&:=-\rho\chi_\k^\top \P\chi_\k+\chi_{\k+1}^\top \tilde\X \chi_{\k+1}+\filty_\k^\top\M\filty_\k\\&\qquad  +\frac {\rho}{\gamma(1-\rho)}\|\z_\k\|^2  -\frac{\rho(\gamma-\mu) }{1-\rho} \|\w_\k\|^2\leq 0, \label{eq:ana_diss2}
  \end{align}
  respectively. Combining both inequalities as follows
  \begin{align*}
  \delta_2(\t) +\sum_{\k=0}^{\t-1} \rho^{\t-\k}\delta_1(\k) \leq 0
  \end{align*}
  and using the telescoping sum argument
  \begin{flalign}\label{eq:tel_sum}
    \sum_{\k=0}^{\t-1}\! \rho^{\t-\k}\! \big(\!\chi_{\k+1}^\top \P\chi_{\k+1}\!-\!\rho\chi_\k^\top\! \P\chi_\k\!\big)\! = \rho \chi_\t^\top\! \P\chi_\t -\rho^{\t+1}\! \underbrace{\chi_0^\top \P\chi_0}_{=0}\hspace{-1cm}\ &&
  \end{flalign}
  leads to 
  \begin{align*}\nonumber
    0&\geq \frac{\rho}{\gamma(1-\rho)} \|\z_\t\|^2-\frac{\rho(\gamma-\mu)}{1-\rho} \|\w_\t\|^2- \mu \sum_{\k=0}^{\t-1}\rho^{\t-\k} \|\w_\k\|^2 \\ 
    &\quad \ +\sum_{\k=0}^\t \rho^{\t-\k}\filty_\k^\top \M\filty_\k+ \chi_{\t+1}^\top \tilde\X \chi_{\t+1}. 
  \end{align*} 
  Note that $\chi_{\t+1}^\top \tilde\X \chi_{\t+1}=\filtx_{\t+1}^\top \X \filtx_{\t+1}$ such that we can use the IQC~\eqref{eq:hardIQC} to obtain
  \begin{align*}\nonumber
    0&\geq \frac{\rho}{\gamma(1-\rho)} \|\z_\t\|^2-\frac{\rho(\gamma-\mu)}{1-\rho} \|\w_\t\|^2- \mu \sum_{\k=0}^{\t-1}\rho^{\t-\k} \|\w_\k\|^2. 
  \end{align*}  
  Further, using the geometric sum upper estimate 
  \begin{align}\label{eq:wpeak}
    \sum_{\k=0}^{\t-1}\rho^{\t-\k} \|\w_\k\|^2 \leq  \rho \|\w\|^2_\peak\sum_{\k=0}^{\t-1}\rho^{\k} \leq \frac{\rho}{1-\rho} \|\w\|^2_\peak
  \end{align}
  it follows 
  \begin{align*}
      0\geq \frac{\rho}{\gamma(1-\rho)} \|\z_\t\|^2 -\frac{\rho(\gamma-\mu)}{1-\rho} \|\w_\t\|^2- \frac{\mu \rho}{1-\rho} \|\w\|^2_\peak.
  \end{align*} 
  Finally, using this inequality divided by $\frac{\rho}{1-\rho}>0$, we obtain
  \begin{align*}
    0 &\geq \frac 1 \gamma\|\z_\t\|^2 -(\gamma-\mu) \|\w_\t\|^2- \mu \|\w\|^2_\peak\\
    &\geq \frac 1 \gamma\|\z_\t\|^2 -(\gamma-\mu) \|\w\|^2_\peak-\mu\|\w\|^2_\peak\\
    &\geq  \frac 1 \gamma\|\z_\t\|^2 -\gamma \|\w\|^2_\peak .
  \end{align*} 
  As the above reasoning holds for all $\t\in\N$ and $\Delta \in\Deltaset$, we deduce $\|\z\|^2_\peak \leq  \gamma^2\|\w\|^2_\peak $, i.e., $\sup_{\Delta\in\Deltaset}\|\DG\|_\indpeak \leq \gamma$.
\end{proof}

\begin{rem}\label{rem:sdp}
  For fixed $\rho$, the matrix inequalities~\eqref{eq:ana_LMI} and~\eqref{eq:ana_LMI2} are linear in the decision variables $(\P, \M, \X, \gamma, \mu)$ and can be solved with an LMI solver.
  Suppose that we have an LMI defining the set $\MXset(\rho)\subseteq \S^{\nfilty}\times \S^{\nfilt}$ such that for all $(\M,\X)\in\MXset(\rho)$, any $\Delta \in \Deltaset$ satisfies the $\rho$-hard IQC defined by $(\rho,\filt,\M,\X)$ (cf. Section~\ref{sec:param} for examples of such LMIs $\MXset(\rho)$).
  Then, we can minimize the upper bound $\gamma^\star (\rho)$ from Theorem~\ref{thm:ana} with fixed $\rho$ by solving the semi-definite program
  \begin{align}\label{eq:sdp}
    \begin{split}
      \gamma^\star(\rho) = \min\ &\gamma \\
      \mathrm{s.t.}\ & ( \M, \X)\in\MXset(\rho), \gamma \geq \mu\geq 0,~\eqref{eq:ana_LMI},~\eqref{eq:ana_LMI2}.
    \end{split}
  \end{align}
  Performing a line search over $\rho\in(0,1)$ to minimize $\gamma^\star(\rho)$ yields the best upper bound on the robust peak-to-peak gain we can get using Theorem~\ref{thm:ana}.
\end{rem}

If the IQC holds pointwise, we can use a multiplier $\M_2\neq \M$ in~\eqref{eq:ana_LMI2} that is different from the multiplier $\M$ in~\eqref{eq:ana_LMI}, as the following theorem shows.
This additional degree of freedom can reduce the conservatism as Example~\ref{exmp:rieber} demonstrates.

\begin{thm}\label{thm:ana_pointwise}
  Assume that there exists a set $\Mset\subseteq \mathbb{S}^{\nfilty}$ such that all $\Delta\in\Deltaset$ satisfy the pointwise IQC defined by $(\filt,\M)$ for all $\M\in\Mset$. Further, assume that there exists $\P\in\S^{n_\chi}$, $\rho\in(0,1)$, $\gamma\geq \mu$, $\mu\geq 0$, $\M\in\Mset$, and $\M_2\in\Mset$ such that~\eqref{eq:ana_LMI} and
  \begin{align}\label{eq:ana_LMI2_pointwise}
    \newcommand{\matspace}{\hspace{0.4cm}}
    \symb{{\arraycolsep=1pt\begin{bmatrix}-\rho\P\\
          &\M_2\\
          &&\frac {\rho}{\gamma(1-\rho)} I\\
          &&&-\frac{\rho(\gamma-\mu) }{1-\rho}I\end{bmatrix}}
      \begin{bmatrix}I&0&0\\[0.3mm]
        \C\all\filty&\D\all\filty\p&\D\all\filty\w\\[0.8mm]
        \C\all\z & \D\all\z\p & \D\all\z\w \\[0.5mm]
        0&0&I\end{bmatrix}}  \preceq 0,
  \end{align}
  hold. Then $\|\DG\|_\indpeak \leq \gamma $ for all $\Delta \in \Deltaset$.
\end{thm}
\begin{proof}
  The proof is analogue to the proof of Theorem~\ref{thm:ana}, since~\eqref{eq:ana_LMI2_pointwise} is equivalent to~\eqref{eq:ana_LMI2} with $\tilde X=0$ and $\M$ replaced by $\M_2$. Then, the only difference is that we directly drop $\filty_k^\top\M\filty_\k$ in~\eqref{eq:ana_diss1} and $\filty_k^\top\M_2\filty_\k$ in~\eqref{eq:ana_diss2} due to the pointwise IQC. 
  The remainder of the proof stays unchanged and yields $\|\DG\|_\indpeak\leq \gamma$.
\end{proof}

Note that both Theorems~\ref{thm:ana} and~\ref{thm:ana_pointwise} do not require positive definiteness of $\P-\tilde \X$ or $\P$.
Hence, we obtain an outer approximation of the reachable set of the output $\z$ but not of the state $\x$ or the extended state $\chi$.
Still these results can be used to obtain an ellipsoidal outer approximation of the reachable set.
In particular, for a given shape matrix $Q\in \S^\nx$, $Q\succ 0$ we can consider the performance output $\z=\C\G\z \x$, where $\C\G\z$ is defined by the Cholesky decomposition $(\C\G\z)^\top \C\G\z = Q$, then the Theorems~\ref{thm:ana} and~\ref{thm:ana_pointwise} yield the following ellipsoidal outer approximation of the reachable set: $\|\x_\k\|_Q \leq \gamma \|\w\|_\peak $ for all $\k\in\N$. 
However, often one does not want to fix the shape matrix $Q$ beforehand, but rather use it as decision variable and fix $\gamma$ and $\mu$ instead.
For such cases, we provide the following theorem.
\begin{thm}\label{thm:rpi}
  Assume that all $\Delta\in\Deltaset$ satisfy the $\rho$-hard IQC defined by $(\rho,\filt,\M,\X)$. Further, let $\mu=1-\rho$ and assume that there exist $\P\in\S^{n_\chi}$ and $Q\in\S^\nx$ such that~\eqref{eq:ana_LMI}, $Q\succ 0$, and $\P \succ \diagmat{\X\\&Q}$ hold.
  Then $\|\x_\k\|_Q^2 \leq\|\chi_\k \|_{\P-\tilde \X}^2 \leq\|\w\|_\peak^2$ holds for all $\k\in\N$ and all $\Delta \in \Deltaset$.
\end{thm}
\begin{proof}
  Due to $\P \succ \diagmat{\X\\&Q}$ it follows $\|\x_\k\|_Q^2 \leq\|\chi_\k \|_{\P-\tilde \X}^2$.
  Since~\eqref{eq:ana_LMI} implies~\eqref{eq:ana_diss1}, we compute $\sum_{k=0}^{\t-1}\rho^{t-k-1}\delta_1(k)\leq 0$ and we use a telescoping sum argument similar to~\eqref{eq:tel_sum} to obtain
  \begin{align*}
    \chi_\t^\top \P \chi_\t
    +\sum_{k=0}^{\t-1}\rho^{t-k-1}\filty_\k^\top \M\filty_\k - \mu \sum_{\k=0}^{\t-1}\rho^{\t-\k-1} \|\w_\k\|^2\leq 0.
  \end{align*}
  Using the $\rho$-hard IQC~\eqref{eq:hardIQC} as well as~\eqref{eq:wpeak} with $\mu=1-\rho$ yields
  \begin{align*}
    \chi_\t^\top \P \chi_\t 
    - \filtx_\t^\top \X \filtx_\t -  \|\w\|^2_\peak \leq 0.
  \end{align*}
  Hence, $\|\chi_\t \|_{\P-\tilde \X}^2= \chi_\t^\top \P \chi_\t - \filtx_\t^\top \X \filtx_\t \leq \|\w\|_\peak^2$ holds for all $\t\in\N$.
\end{proof}
 In the special case of $\rho$-hard IQCs with $\X=0$, Theorem~\ref{thm:rpi} is a corollary of~\cite[Theorem~2]{Schwenkel2020}.

\section{$\rho$-hard IQC for parametric uncertainties}\label{sec:param}
In this section, we provide $\rho$-hard IQCs with terminal conditions for the uncertainty class of parametric uncertainties.
Thereby, we demonstrate how we can construct rich sets $\MXset(\rho)$ using LMIs such that we can implement the procedure from Remark~\ref{rem:sdp}.
Furthermore, we show that our approach unifies existing LMI approaches to determine the robust peak-to-peak gain.
In particular, we consider the following the sets of parametric uncertainties $\Deltaset_{\mathrm{p}} = \{\ell_{2\mathrm{e}}^\nq\to\ell_{2\mathrm{e}}^\np\mid \Delta(\q)_\k = \Delta_k\q_k\}$:
\begin{itemize}
  \item time-varying with coefficients from a polytope
\end{itemize}
  \begin{flalign}
    \Deltaset_{\mathrm{p,tv}}:=\{\Delta\in\Deltaset_{\mathrm{p}} \mid \Delta_\k = \diag(\delta_k), \delta_k \in \mathrm{conv} \{\delta^{1}, \dots, \delta^{m}\} \}\hspace{-1cm}&& \label{eq:Delta_ptv}
  \end{flalign} 
\begin{itemize}
  \item[] with $\delta^{1}, \dots, \delta^{m}\in \R^{\nq}$ as in~\cite{Rieber2008},
  \item time-invariant with coefficients from a polytope
\end{itemize}
  \begin{align}\label{eq:Delta_pti}
    \Deltaset_{\mathrm{p,ti}}:=\{\Delta\in\Deltaset_{\mathrm{p}} \mid \Delta_\k = \diag(\delta), \delta \in \mathrm{conv} \{\delta^{1}, \dots, \delta^{m}\} \}
  \end{align} 
\begin{itemize}
  \item[] with $\delta^{1}, \dots, \delta^{m}\in \R^{\nq}$ as in~\cite{Rieber2008},
  \item time-varying with gain bound as in \cite{Ji2007}:
\end{itemize}
  \begin{align}
    \Deltaset_{\mathrm{p,tv,g}}:=\{\Delta\in\Deltaset_{\mathrm{p}} \mid \Delta_k^\top \Delta_k \leq I  \}.\label{eq:Delta_ptvg}
  \end{align}
These sets of uncertainties have been studied within the classical IQC framework leading to well-known IQCs (see, e.g., \cite{Megretski1997} or \cite{Veenman2016}).
The class of IQCs~\cite[Class~4]{Veenman2016} for $\Deltaset_{\mathrm{p,tv}}$ is known to hold pointwise and thus we can directly use it without modifications.

\begin{thm}\!(Pointwise IQC for $\Deltaset_{\mathrm{p,tv}}$). \ \ \label{thm:param_iqc_tv}
  Define $\filt=I_{2\nq}$ and $\Mset$ as the set of all $\M\in\S^{2\nq}$ that satisfy
  \begin{align*}
    \symb{\M \begin{bmatrix}0\\ I_{\nq}\end{bmatrix}}\preceq 0 \ \text{and}\ \symb{\M \begin{bmatrix} I_\nq \\ \diag(\delta^j)\end{bmatrix}} \succeq 0 \ \forall j=1,\dots,m.
  \end{align*}
  Then any $\Delta\in\Deltaset_{\mathrm{p,tv}}$ satisfies the pointwise IQC defined by $(\filt, \M)$ for all $\M\in\Mset$.
\end{thm}
\begin{proof}
  The proof in~\cite[Class~4]{Veenman2016} can be analogously transferred to discrete-time.
\end{proof}
The class of IQCs~\cite[Class~6]{Veenman2016} for $\Deltaset_{\mathrm{p,ti}}$ on the other hand is a class of soft IQCs and hence, we cannot use it directly but need to restrict it to $\rho$-hard IQCs with terminal cost to apply Theorem~\ref{thm:ana}.
\begin{thm}\!($\rho$-hard IQC with terminal cost for $\Deltaset_{\mathrm{p,ti}}$). \label{thm:param_iqc}\ \ 
  Let $\Phi=\ssrep{\A\Phi}{\B\Phi{}}{\C\Phi{}}{\D\Phi{}{}}\in\RHinf^{n_\sigma \times 1}$ have state dimension $n_\phi$, i.e., $\A\Phi\in\R^{n_\phi \times n_\phi}$. 
  Define $\filt$ as the Kronecker product $\filt = I_{2\nq}\otimes \Phi$, i.e., 
  \begin{align*}
    \ssrepflex{c}{c:c}{\A\filt & \B\filt\q & \B\filt\p}{\C\filt\filty & \D\filt\filty\q & \D\filt\filty\p} = \ssrepflex{cc}{c:c}{I_{\nq}\otimes\A\Phi & 0 & I_{\nq}\otimes \B\Phi{} & 0 \\ 0 & I_{\nq}\otimes\A\Phi & 0 & I_{\nq}\otimes \B\Phi{}}{I_{\nq}\otimes\C\Phi{} & 0 & I_{\nq}\otimes \D\Phi{}{} & 0 \\ 0 & I_{\nq}\otimes\C\Phi & 0 & I_{\nq}\otimes \D\Phi{}{}}
  \end{align*}
  and for $\rho\in(0,1)$ define $\MXset(\rho)$ as the set of all $(\M,\X)\in \S^{2n_\sigma\nq} \times \S^{2n_\phi\nq}$ for which there exist $Y^j\in \S^{2n_\phi\nq}$ such that the following matrix inequalities hold for all $j=1,\dots,m$:
  \begin{align}\label{eq:paramIQCLMI1}
    \symb{\diagmat{-\rho \X_{22}\\&\X_{22}\\&& \M_{22}}\begin{bmatrix} I_{n_\phi \nq} & 0 \\ I_{\nq}\otimes \A\Phi & I_{\nq}\otimes \B\Phi{} \\
        I_{\nq}\otimes\C\Phi{} & I_{\nq}\otimes \D\Phi{}{} \end{bmatrix}} &\preceq 0 \\ \label{eq:paramIQCLMI2}
    \symb{\diagmat{-\rho Y^j\\&Y^j\\&& \M}\begin{bmatrix} I_{2n_\phi \nq} & 0 \\ \A\filt & \B\filt\q+\B\filt\p \diag (\delta^j) \\[0.7mm]
        \C\filt\filty & \D\filt\filty\q + \D\filt\filty\p \diag (\delta^j) \end{bmatrix}} &\succeq 0 \\  \label{eq:paramIQCLMI3}
    Y^j-\X &\preceq 0.
  \end{align}
  Then for all $\rho\in(0,1)$ and all $(\M,\X)\in\MXset(\rho)$, any $\Delta\in\Deltaset_{\mathrm{p,ti}}$ satisfies the $\rho$-hard IQC defined by $(\rho,\filt,\M,\X)$.
\end{thm}
\begin{proof}
  We introduce the notation $\phi^i = \ssrep{\A\Phi}{\B\Phi{}}{I}{0} \q^i$ and $\sigma^i = \Phi \q^i$ for state and output of $\Phi$ when the input $\q^i$ is the $i$th component in $\q$.
  Further, we define the block diagonal matrices $\hat\phi =\diagmat{\phi^1 \\[-2mm] & \ddots \\[-1mm] && \phi^{\nq}}$, $\hat\sigma =\diagmat{\sigma^1 \\[-2mm] & \ddots \\[-1mm] && \sigma^{\nq}}$, and $\hat\q=\diagmat{\q^1 \\[-2mm] & \ddots \\[-1mm] && \q^{\nq}}$.
  For clarity, we explicitly highlight the dependency on $\delta$ in the state $\filtx=\filtx(\delta)$ and the output $\filty=\filty(\delta)$ of $\filt$ with the input $\begin{bmatrix}\q \\ \p\end{bmatrix}$, where $\p=\diag(\delta)\q = \hat \q \delta$. 
  Note that the special structure $\filt = I_{2\nq}\otimes \Phi$ implies $\filtx(\delta) =\begin{bmatrix}\hat \phi \boldsymbol{1}_\nq \\ \hat \phi \delta \end{bmatrix}$ and $\filty(\delta) =\begin{bmatrix}\hat \sigma\boldsymbol{1}_\nq \\ \hat \sigma\delta \end{bmatrix}$ where $\boldsymbol{1}_\nq \in\R^{\nq}$ denotes the all ones vector.
  For each vertex $\delta^j$, $j=1,\dots,m$, let us multiply~\eqref{eq:paramIQCLMI2} from the left by $\begin{bmatrix}\filtx_\k(\delta^j)^\top & \q_\k^\top \end{bmatrix}$ and from the right by its transpose, then we obtain
  \begin{align}
    \symb{\diagmat{-\rho Y^j\\&Y^j\\&& \M}\begin{bmatrix} \filtx_\k(\delta^j) \\ \filtx_{\k+1}(\delta^j) \\ \filty_\k(\delta^j) \end{bmatrix}} \geq 0 
  \end{align}
  and thus
  \begin{flalign}
    -\rho\! \symb{\! Y^j \filtx_\k(\delta^j)} + \symb{\! Y^j \filtx_{\k+1}(\delta^j)} + \symb{\! \M \filty_\k(\delta^j)} \geq 0.\hspace{-1cm}&&
  \end{flalign}
  After multiplying this inequality by $\rho^{\t-\k}$, summing it from $\k=0$ to $\k=\t$, and using a telescoping sum argument similar to~\eqref{eq:tel_sum}, it follows
  \begin{align}\label{eq:iqc_delta_j}
    \symb{ Y^j \filtx_{\t+1}(\delta^j)} +\sum_{k=0}^{t} \rho^{\t-k}\symb{\M \filty_k(\delta^j)}\geq 0.
  \end{align}
  Furthermore, let us multiply $\begin{bmatrix}\hat \phi_\k^\top&\hat\q_\k^\top \end{bmatrix}$ to the left and its transpose to the right of~\eqref{eq:paramIQCLMI1}, then we obtain
  \begin{align*}
    -\rho\symb{\X_{22}\hat\phi_\k} + \symb{\X_{22} \hat\phi_{\k+1}} + \symb{\M_{22} \hat\sigma_\k} \preceq 0
  \end{align*}
  and hence, after multiplying this inequality by $\rho^{\t-\k}$, summing it from $\k=0$ to $\k=\t$, and using a telescoping sum argument similar to~\eqref{eq:tel_sum}, it follows
  \begin{align}\label{eq:concave}
    \sum_{k=0}^{t} \rho^{\t-\k} \symb{\M_{22} \hat\sigma_\k} +\symb{\X_{22} \hat\phi_{\t+1}}\preceq 0.
  \end{align}
  Since both $\filtx=\filtx(\delta)$ and $\filty=\filty(\delta)$ are linear in $\delta$, we can write
  \begin{flalign}\label{eq:iqc_afo_delta}
    &\sum_{k=0}^{t}\! \rho^{\t-k}\symb{\!\M \filty_k(\delta)}+ \symb{\!\X \filtx_{\t+1}(\delta)} \!
    = c_0  
    + c_1^\top \delta 
    +\delta^\top\! c_2 \delta \hspace{-1cm} & 
  \end{flalign}
  with suitable $c_0\in\R$, $c_1\in\R^\nq$, and $c_2 \in \R^{\nq\times\nq}$ which are independent of $\delta$.
  Considering  $\filtx(\delta) =\begin{bmatrix}\hat \phi \boldsymbol{1}_\nq \\ \hat \phi \delta \end{bmatrix}$ and $\filty(\delta) =\begin{bmatrix}\hat \sigma\boldsymbol{1}_\nq \\ \hat \sigma\delta \end{bmatrix}$, we can compute $c_2$ as
  \begin{align*}
    c_2 =\sum_{k=0}^{t} \rho^{\t-\k} \symb{\M_{22} \hat\sigma_\k} +\symb{\X_{22} \hat\phi_{\t+1}},
  \end{align*}
  and from~\eqref{eq:concave} we know that $c_2 \preceq 0$.
  This means that~\eqref{eq:iqc_afo_delta} is a concave function in $\delta$ and hence we can lower bound it by the convex combination $\delta = \sum_{j=1}^m\theta_j\delta^j$ with $\theta_j\geq 0$ and $\sum_{j=1}^m \theta_j = 1$ as follows
  \begin{align*}
    &\sum_{k=0}^{t} \rho^{\t-k}\symb{\M \filty_k(\delta)}+ \symb{\X \filtx_{\t+1}(\delta)} \\
    &\geq \sum_{j=1}^{m}\theta_j\left(\sum_{k=0}^{t} \rho^{\t-k}\symb{ \M \filty_k(\delta^j)}+ \symb{ \X \filtx_{\t+1}(\delta^j)}\right)\\
    &\refeq{\eqref{eq:paramIQCLMI3}}{\geq} \sum_{j=1}^{m}\theta_j\left(\sum_{k=0}^{t} \rho^{\t-k}\symb{\M \filty_k(\delta^j)}+ \symb{Y^j \filtx_{\t+1}(\delta^j)}\right) \refeq{\eqref{eq:iqc_delta_j}}{\geq} 0.
  \end{align*}
  This shows~\eqref{eq:hardIQC} and completes the proof as the arguments hold for all $\delta \in \mathrm{conv}\{\delta^1,\dots,\delta^m\}$, i.e., all $\Delta \in \Deltaset$, all $\q\in\ell_{2\mathrm{e}}^\nq$, all $\t\geq 0$, and all $\rho,\M,\X$ satisfying~\eqref{eq:paramIQCLMI1}--\eqref{eq:paramIQCLMI3}. 
\end{proof}

In Theorem~\ref{thm:param_iqc}, we added the LMI~\eqref{eq:paramIQCLMI3} compared to the soft IQC from~\citep[Class 6]{Veenman2016} in order to make the IQC hard and we introduced $\rho$ in~\eqref{eq:paramIQCLMI1} and~\eqref{eq:paramIQCLMI2} to make it $\rho$-hard.
\begin{rem}\label{rem:phi}
  A standard choice (cf. \citep{Veenman2016}) for the system $\Phi= \ssrep{\A\Phi}{\B\Phi{}}{\C\Phi{}}{\D\Phi{}{}}$ is $\A\Phi = \lambda I_\greeknu + \begin{bmatrix}0&0\\ I_{\greeknu-1}&0\end{bmatrix}\in\R^{\greeknu\times \greeknu}$, $\B\Phi{}=\begin{bmatrix}1 \\ 0\end{bmatrix}\in\R^\greeknu$, $\C\Phi{} = \begin{bmatrix}0 \\ I_\greeknu \end{bmatrix}\in\R^{\greeknu+1 \times \greeknu}$, and $\D\Phi{}{}=\begin{bmatrix}1 \\ 0\end{bmatrix}\in\R^{\greeknu+1}$ for some $\lambda \in (-1,1)$ and for some natural number $\greeknu$ leading to the dimensions $n_\phi=\nu$, $n_\sigma = \greeknu+1$.
  Then, one can decide on an order $\greeknu$ and perform a second line search over $\lambda$ to find the $\rho(\lambda)$ that leads to the best upper bound $\gamma^\star (\rho(\lambda))$ in~\eqref{eq:sdp}.
\end{rem}
\begin{rem}\label{rem:rieber}
  If we use Theorem~\ref{thm:ana_pointwise} together with the pointwise IQC from Theorem~\ref{thm:param_iqc_tv} for $\Delta\in\Deltaset_{\mathrm{p,tv}}$, then we recover the conditions used in \citep[Theorem~4, Remark~5]{Rieber2008} as a special case. In particular, the LMIs~\eqref{eq:ana_LMI},~\eqref{eq:ana_LMI2_pointwise}, $\M=\M_1\in\Mset$, and $\M_2\in\Mset$ correspond in this order to \citep[(27),~(28),~(25), and~(26)]{Rieber2008} with the necessary modifications in the notation. 
  In the time-invariant case $\Delta\in\Deltaset_{\mathrm{p,ti}}$, the approach in~\citep{Rieber2008} cannot exploit the time invariance and uses the same bounds as for $\Delta\in\Deltaset_{\mathrm{p,tv}}$, whereas with the IQC approach, one can use the dynamic IQC from Theorem~\ref{thm:param_iqc} to obtain improved bounds on $\gamma$ (cf. the example in Section~\ref{sec:exmp}).
\end{rem}

Finally, for $\Deltaset_{\mathrm{p,tv,g}}$, we can use the pointwise IQC from~\citep[Class~11]{Veenman2016} with $\alpha=1$.

\begin{thm}\!(Pointwise IQC for $\Deltaset_{\mathrm{p,tv}}$). \ \ \label{thm:param_iqc_tv_g}
  Define $\filt=I_{\nq+\np}$ and $\Mset=\{\varepsilon\diag(I_{\nq},-I_\np)\mid \varepsilon\geq 0\}$.
  Then all $\Delta\in\Deltaset_{\mathrm{p,tv,g}}$ satisfy the pointwise IQC defined by $(\filt, \M)$ for all $\M\in\Mset$.
\end{thm}
\begin{proof}
  Due to $\filt=I_{\nq+\np}$ and $\M\in\Mset$ it is $\filty_k=\begin{bmatrix}\q_\k^\top & \p_\k^\top \end{bmatrix}^\top$ and
  \begin{align*}
    \symb{\M \begin{bmatrix}\q_k \\ \p_k \end{bmatrix}}&=\symb{\diagmat{\varepsilon I_\nq \\ &-\varepsilon I_\np} \begin{bmatrix}\q_k \\ \Delta_\k \q_k \end{bmatrix}}\\&=\q_k^\top\varepsilon( I_\nq-\Delta_k^\top \Delta_k) \q_k\geq 0.\\[-0.8cm]
  \end{align*}
\end{proof}
\begin{rem}
  Our problem setup with $\Delta\in\Deltaset_{\mathrm{p,tv,g}}$ is equivalent to the problem considered in~\citep{Ji2007}, if we set $\Delta_k=F(k)$ $\C\G\q=N_a$, $\B\G\p=M$, (the $F(k)$, $N_a$, and $M$ from~\citep{Ji2007}), $\D\G\q\p=0$, and $\D\G\q\w=0$.
  If we use Theorem~\ref{thm:ana_pointwise} together with the pointwise IQC from Theorem~\ref{thm:param_iqc_tv_g} and fix $\M_2=0$ and $\mu=1-\rho$, we obtain the equivalent conditions as used in \citep[Theorem~3]{Ji2007} to compute an upper bound on the robust peak-to-peak gain. In particular, the LMIs~\eqref{eq:ana_LMI},~\eqref{eq:ana_LMI2_pointwise} correspond to \citep[(12),~(14)]{Ji2007} with the necessary transformations and modifications in the notation.
  Hence, in the special case of $\Delta\in\Deltaset_{\mathrm{p,tv,g}}$ our approach reduces to~\citep{Ji2007}.
\end{rem}

It is important to note that many more types of uncertainties can be characterized with the help of $\rho$-hard IQCs, see for example~\citep{Schwenkel2022a} for a pointwise IQC for uncertain time delays or~\cite{Boczar2019} for $\rho$-hard IQCs on sector and slope restricted nonlinearities.
\section{Numerical Examples}\label{sec:exmp}
In the first example, we compute the robust peak-to-peak gain and compare Theorem~\ref{thm:ana}, Theorem~\ref{thm:ana_pointwise}, and the approach in~\citep{Rieber2008}.
In the second example, we perform a reachability analysis and compare Theorem~\ref{thm:rpi} to the approach in~\citep{Jaoude2021}.
\begin{exmp}\label{exmp:rieber}
  We consider the following MIMO system $\G$ taken from~\citep{Rieber2008}
  \begin{align*}
    \left[\begin{array}{c:c:c}\A\G & \B\G\p & \B\G\w \\[0.5mm] \hdashline\rule{0pt}{3.5mm} \C\G\q & \D\G{\q}\p & \D\G\q\w\\[0.5mm] \hdashline\rule{0pt}{3.5mm} \C\G\z & \D\G\z\p & \D\G\z\w \end{array}\right] = 
    \left[\begin{array}{cc:cc:cc}0.2 & 0.01 & 0.1 & 0.2 & 3 & 2 \\
      -0.1 & -0.01 & 0.3 & -0.2 & 3 & 1 \\ \hdashline\rule{0pt}{3.5mm}
      0.2 & -0.3 & 0.4 & 0.3 & 3 & 1 \\
      0.8 & 0.5 & -0.6 & 0.1 & 2 & 7 \\ \hdashline\rule{0pt}{3.5mm} 
      2 & 1 & 1 & 2 & 1 & -2 \\
      2 & 3 & -1 & 4 & -4 & 3 \end{array}\right],
  \end{align*}
  which is in feedback with $\Delta \in \Deltaset_\mathrm{p,ti}$ according to~\eqref{eq:sys},~\eqref{eq:Delta}, with $\Deltaset_\mathrm{p,ti}$ defined in~\eqref{eq:Delta_pti} with
  \begin{align*}
    \delta^1=\begin{bmatrix}-0.1\\ -0.3\end{bmatrix}, \ \delta^2=\begin{bmatrix}-0.1\\ 0.6\end{bmatrix},\ \delta^3=\begin{bmatrix}0.5\\ -0.3\end{bmatrix}, \ \delta^4=\begin{bmatrix}0.5\\ 0.6\end{bmatrix}.
  \end{align*}
  To compute an upper bound on the peak-to-peak gain, we use the analysis in Theorem~\ref{thm:ana} in combination with the IQC from Theorem~\ref{thm:param_iqc} with $\Phi$ from Remark~\ref{rem:phi} with $\greeknu=2$ and $\lambda=-0.25$.
  Then, we solve the semi-definite program~\eqref{eq:sdp} with the LMI parser YALMIP~\citep{Lofberg2004} and the solver~\cite{mosek} and obtain the bound $\gamma \leq 60.61$ for $\rho = 0.18$.
  For comparison, we found the lower bound $\gamma \geq 59.41$ on the peak-to-peak gain by maximizing the peak of the output $\|\z\|_\peak\leq \gamma $ over $\Delta\in\Deltaset_{\mathrm{p,ti}}$ and $\w$ with $\|\w\|_\peak \leq 1$.
  This maximization problem is non-convex and hence we cannot expect that the local maximum we found results in a tight lower bound.
  The local maximum was obtained by using MATLAB's \texttt{fmincon} and starting the optimization at the initial guess $\delta_1=0.5$, $\delta_2=0.6$ and $\w_\k=\frac {1}{\sqrt{2}}\begin{bmatrix}1 & 1\end{bmatrix}^\top $ for all $\k\in[0,8]$.
  If we instead use the result in~\citep{Rieber2008}\footnote{Specifically Theorem~4 in combination with Remark~5} to compute an upper bound $\gamma$ on the peak-to-peak gain, then we obtain $\gamma \leq 66.93$, which shows that the proposed IQC-based improves the state-of-the-art approach significantly.
  The reason for this significant improvement lies in the fact, that the approach of~\citep{Rieber2008} is equivalent to using Theorem~\ref{thm:ana_pointwise} with the pointwise IQC from Theorem~\ref{thm:param_iqc_tv} for $\Delta\in\Deltaset_{\mathrm{p,tv}}$ (cf. Remark~\ref{rem:rieber}), i.e., the knowledge about the time-invariance of $\Delta$ is not exploited in their approach. 
  So if we allow the uncertain parameters to be time varying, i.e., $\Delta\in\Deltaset_{\mathrm{p,tv}}$ as defined in~\eqref{eq:Delta_ptv}, then we obtain with both approaches the same upper bound $\gamma \leq 66.93$.
  We can also use Theorem~\ref{thm:ana} in combination with the pointwise IQC from Theorem~\ref{thm:param_iqc_tv}, as pointwise IQCs imply $\rho$-hard IQCs with $\X=0$ for all $\rho \in (0,1)$. 
  Then, we obtain the best resulting bound $\gamma\leq 67.81$ for $\rho=0.23$. 
  This shows that Theorem~\ref{thm:ana_pointwise} indeed yields less conservative results than Theorem~\ref{thm:ana} if the IQC holds pointwise. 
  Again, for comparison we computed a lower bound $\gamma \geq 65.66$ using MATLAB's \texttt{fmincon} and starting the optimization at the initial guess of our previously found time-invariant worst case.
  All obtained bounds on $\gamma$ are summarized in Table~\ref{tab:gamma}.
  
  \begin{table}
    \begin{tabular}{ccccc}
      \toprule
      & \!\!\cite{Rieber2008}\!\! & Theorem~\ref{thm:ana} & Theorem~\ref{thm:ana_pointwise} & lower bound \\ \midrule
      $\Delta \in \Deltaset_{\mathrm{p,ti}}$ & $\gamma \leq 66.93$ &  $\gamma \leq 60.61$ &  $\gamma \leq 66.93$ &  $\gamma \geq 59.41$ \\
      $\Delta \in \Deltaset_{\mathrm{p,tv}}$ & $\gamma \leq 66.93$ &  $\gamma \leq  67.81$ &  $\gamma \leq 66.93$ &  $\gamma \geq 65.66$ \\\bottomrule
    \end{tabular}\\[0.01cm]
    \caption{\makebox{Best obtained bounds on $\gamma$.}}\label{tab:gamma}
  \end{table}
\end{exmp}
\begin{exmp}\label{exmp:jaude}
  We consider the following MIMO system $\G$ taken from~\citep{Jaoude2021}
  \begin{align*}
    \left[\begin{array}{c:c:c}\A\G & \B\G\p & \B\G\w \\[0.5mm] \hdashline\rule{0pt}{3.5mm} \C\G\q & \D\G{\q}\p & \D\G\q\w \end{array}\right] = 
    \left[\begin{array}{ccc:cc:cc}.05 & -.2 & .3 & .2 & .1 & .5 & .1 \\
      .1 & .8 & .2 & .5 & -.4 & -.3 & -.7 \\
      -.2 & .5 & -.1 & -.3 & -.2 & .5 & -.2 \\ \hdashline\rule{0pt}{3.5mm}
      1 & -.5 & .3 & .1 & .6 & -.5 & .4 \\
      .9 & .2 & -.5 & .6 & -.9 & .3 & .1 \end{array}\right],
  \end{align*}
  which is in feedback with $\Delta \in \Deltaset_{\mathrm{p,ti}}$ with $\delta^1 = \begin{bmatrix}-0.3 & -0.3\end{bmatrix}^\top$ and $\delta^2 = \begin{bmatrix}0.3 & 0.3\end{bmatrix}^\top$ and is subject to disturbances $\|\w\|_\infty\leq 0.5$, which implies $\|\w\|_\peak\leq \sqrt{0.5}$. 
  Using Theorem~\ref{thm:rpi} with the $\rho$-hard IQC from Theorem~\ref{thm:param_iqc} and $\Phi$ from Remark~\ref{rem:phi} with $\greeknu=2$ and $\lambda =0.2$, we can compute $Q\in\S^\nx$ such that $\|\x_\k\|_{\tilde Q}\leq 1$ for all $\k\in\N$ with $\tilde Q=Q\|\w\|_\peak^{-2}$.
  We minimize the volume of the ellipsoidal outer approximation by minimizing $-\log\det (Q)$, which yields\footnote{To avoid numerical problems, we rescaled LMI~\eqref{eq:ana_LMI} by using $\hat \mu=1000 \mu$ and restored $Q=\frac{1}{1000} \hat Q$.} $-\log\det \left(\tilde Q \right) = 7.04$ for $\rho=0.96$.
  \cite{Jaoude2021} used an approach based on pointwise IQCs to obtain $-\log\det \left(\tilde Q \right) = 8.38$.
  Note that $\|\w\|_\peak\leq \sqrt{0.5}$ is a conservative bound on  $\|\w\|_\infty\leq 0.5$ and yet, the proposed approach using $\rho$-hard IQCs with a terminal cost and $\|\w\|_\peak\leq \sqrt{0.5}$ can reduce the conservatism significantly compared to using pointwise IQCs and $\|\w\|_\infty\leq 0.5$.
\end{exmp}
\section{Conclusion and outlook}
In this work we presented a framework to guarantee an upper bound on the robust peak-to-peak gain of a linear system that is in feedback with an uncertain component.
The key tool we used in our analysis were $\rho$-hard IQCs with a terminal cost.
We conjecture that analogue results may hold in continuous time. 
Future work concern robust controller \emph{synthesis} for peak-to-peak gain minimization.

{\small \bibliography{my_bib}}
\end{document}